\documentclass[sigplan, 10pt,   nonacm]{acmart}

    \usepackage[british]{babel}
    \usepackage{mathpartir}
    \usepackage{thmtools}
    \usepackage{thm-restate}
    \usepackage{ifthen}
    \usepackage{url}        
    \usepackage{doi}

    \usepackage{genMath}
    \usepackage{collections}
    \usepackage{separationLogic}
    \usepackage{obligations}

    \usepackage{syntax}
    \usepackage{threadPools}
    \usepackage{variables}
    \usepackage{reduction}
    \usepackage{safety}
    \usepackage{proofRules}

\newtheorem{mytheorem}{Theorem}[section]
\newtheorem{mydefinition}[mytheorem]{Definition}
\newtheorem{mylemma}[mytheorem]{Lemma}

\title{A Separation Logic to Verify Termination of Busy-Waiting for Abrupt Program Exit}

\author{Tobias Reinhard}
\affiliation{
    \institution{KU Leuven\\imec-DistriNet Research Group}
}
\email{tobias.reinhard@kuleuven.be}

\author{Amin Timany}
\affiliation{
    \institution{Aarhus University\\Logic and Semantics Group}
}
\email{timany@cs.au.dk}

\author{Bart Jacobs}
\affiliation{
    \institution{KU Leuven\\imec-DistriNet Research Group}
}
\email{bart.jacobs@kuleuven.be}

\begin{document}

    \NewDocumentCommand{\RRedRuleNameSet}{}{\ensuremath{N_\keywordFont{r}}\xspace}
    \NewDocumentCommand{\progOrdGraph}{}{\ensuremath{\mathcal{G}}\xspace}
    \NewDocumentCommand{\progOrdGraphConfig}{}{\ensuremath{g}\xspace}
    \NewDocumentCommand{\progOrdGraphLF}{}
            {\ensuremath{\progOrdGraph_{\fixedPredNameFont{lf}}}\xspace}
    \NewDocumentCommand{\nodesOf}{m}
            {\ensuremath{\fixedFuncNameFont{nodes}(#1)}\xspace}
    \NewDocumentCommand{\edgesOf}{m}
            {\ensuremath{\fixedFuncNameFont{edges}(#1)}\xspace}
    \NewDocumentCommand{\leavesOf}{m}
            {\ensuremath{\fixedFuncNameFont{leaves}(#1)}\xspace}

    \NewDocumentCommand{\sumOb}{m}{\ensuremath{S_o(#1)}\xspace}
    \NewDocumentCommand{\sumCred}{m}{\ensuremath{S_c(#1)}\xspace}

\begin{abstract}
    Programs for multiprocessor machines commonly perform busy-waiting for synchronisation. 
    In this paper, we make a first step towards proving termination of such programs. 
    We approximate 
    (i)~arbitrary waitable events by abrupt program termination and 
    (ii)~busy-waiting for events by busy-waiting to be abruptly terminated.

    We propose a separation logic for modularly verifying termination (under fair scheduling) of programs where some threads eventually abruptly terminate the program, and other threads busy-wait for this to happen.
\end{abstract}

\maketitle

\section{Introduction}\label{sec:Introduction}

    Programs for multiprocessor machines commonly perform busy-waiting for synchronisation~\cite{Mhlemann1980MethodFR, MellorCrummey1991AlgorithmsFS, Rahman2012ProcessSI}. 
    In this paper, we make a first step towards proving termination of such programs. 
    Specifically, we propose a separation logic~\cite{Reynolds2002SeparationLA, OHearn2001LocalRA} for modularly verifying termination (under fair scheduling) of programs where some threads eventually abruptly terminate the program, and others busy-wait for this to happen.

    Here, by modular we mean that we reason about each thread and each function in isolation.
    That is, we do not reason about thread scheduling or interleavings.
    We only consider these issues when proving the soundness of our logic.
    
    In this work, we approximate 
    (i)~arbitrary events that a program might wait for by abrupt termination and 
    (ii)~busy-waiting for events by busy-waiting to be abruptly terminated. 
    In Section~\ref{sec:FutureWork}, we sketch preliminary ideas for generalizing this to verifying termination of busy-waiting for arbitrary events, and how this work may also be directly relevant to verifying liveness properties of a program’s I/O behaviour.
    
    Throughout this paper we use a very simple programming language to illustrate our verification approach.
    Its simplicity would allow us to verify termination of busy-waiting for abrupt termination via a static analysis significantly simpler than the proposed separation logic.
    However, in contrast to such an analysis, our approach is also applicable to realistic languages.
    We are confident that the logic we propose can be combined straightforwardly with existing concurrent separation logics like Iris~\cite{Jung2018IrisFT} to verify termination of realistic programs where threads busy-wait for abrupt program termination.

    We start by introducing the programming language in Section~\ref{sec:Language} and continue in Section~\ref{sec:Logic} with presenting the separation logic and so-called \emph{obligations} and \emph{credits}~\cite{Hamin2019TransferringOT, Hamin2018DeadlockFreeM, Leino2010DeadlockFreeCA, Kobayashi2006ANT}, which we use to reason about termination of busy-waiting.
    In Section \ref{sec:ProofSystem} we present our verification approach in the form of a set of proof rules and illustrate their application.
    Afterwards, we sketch the soundness proof of our proof system in Section~\ref{sec:Soundness}.
    We conclude by outlining our plans for future work, comparing our approach to related work and reflecting on our approach in Sections~\ref{sec:FutureWork}, \ref{sec:RelatedWork} and ~\ref{sec:Conclusion}.

\section{The Language}\label{sec:Language}
    We consider a simple programming language with an \cmdExit command that abruptly terminates all running threads, a \cmdFork command, a looping construct \cmdLoop[\cmdSkip] to express infinite busy-waiting loops and sequencing $\cmdVar_1; \cmdVar_2$.
    \begin{mydefinition}[Commands and Continuations]
        We denote the sets of commands \cmdVar and continuations \contVar as defined by the grammar presented in Figure~\ref{fig:Syntax} by~\CmdSet and~\ContSet.
        We consider sequencing $\cdot\mathop{;}\cdot$ as defined in the grammars of commands and continuations to be right-associative.
    \end{mydefinition}
    
    We use commands and continuations to represent programs and single threads, respectively, as well as natural numbers for thread IDs.
    Continuation \contDone marks the end of a thread's execution.
    We consider thread pools to be functions mapping a finite set of thread IDs to continuations.
        
    \begin{figure}
        $$
        \begin{array}{r c l}
            \cmdVar \in \CmdSet
                    &::=&
                            \cmdExit ~|~
                            \cmdLoop[\cmdSkip] ~|~
                            \cmdFork[\cmdVar] ~|~
                            \cmdVar\ ;\, \cmdVar
            \\
            \contVar \in \ContSet
                    &::=&
                            \contDone   ~|~   \cmdVar\ ;\, \contVar
        \end{array}    
        $$
        \vspace{-0.5cm}
        \caption{Syntax}
        \label{fig:Syntax}
    \end{figure}
    
    \begin{mydefinition}[Thread Pools]
            We define the set of thread pools \ThreadPoolSet as follows
            $$
                \ThreadPoolSet\ := \
                        \{  \tpVar : \tpDomVar \rightarrow \ContSet  
                                \, \ | \, \
                                \tpDomVar \finSubset \N
                        \}.
            $$
            We denote thread pools by \tpVar, thread IDs by \tidVar and the empty thread pool by 
            $\tpEmpty: \emptyset\rightarrow\ContSet$.
        \end{mydefinition}

    \begin{mydefinition}[Thread Pool Extension]\label{def:ThreadPoolExtension}
        Let $\tpVar:\tpDomVar\rightarrow\ContSet \in \ThreadPoolSet$ be a thread pool. We define:
        \begin{itemize}
            \item $\tpVar \tpExt \emptyset := \tpVar$,
            \item 
                    $\tpVar \tpExt \setOf{\contVar} :
                        \tpDomVar\,\cup\,\setOf{\max(\tpDomVar)+1}\rightarrow\ContSet$ 
                    with\\ 
                    $(\tpVar\tpExt\setOf{\contVar})(\tidVar) = P(\tidVar)$ for all $\tidVar\in\tpDomVar$ and \\
             $(\tpVar\tpExt\setOf{\contVar})(\max(\tpDomVar)+1) = \contVar$ ,
            \item 
                    $\tpVar \tpRem \tidVar' : 
                        \tpDomVar \setminus \setOf{\tidVar'}\rightarrow \ContSet$
                    with\\
                    $(\tpVar \tpRem \tidVar')(\tidVar) = \tpVar(\tidVar)$
        \end{itemize}
    \end{mydefinition}
    
    We consider a standard small-step operational semantics for our language defined in terms of two reduction relations: (i)~\stRedStepSymb for single thread reduction steps and (ii)~\tpRedStepSymb for thread pool reduction steps.
    
    \begin{mydefinition}[Single-Thread Reduction Relation]
        We define a \emph{single-thread reduction relation} \stRedStepSymb according to the rules presented in Figure~\ref{fig:SingleThreadReductionRelation}.
        A reduction step has the form
        $$\stRedStep{\contVar}{\nextContVar}[\ftsVar]$$
        for a set of forked threads $\ftsVar \subset \ContSet$ with $\cardinalityOf{\ftsVar} \leq 1$.
    \end{mydefinition}

    \begin{figure}
        \begin{mathpar}
            \RedSTLoop
            \and
            \RedSTFork
            \and
            \RedSTSeq
        \end{mathpar}
        \vspace{-0.5cm}
        \caption{Reduction rules for single threads.}
        \label{fig:SingleThreadReductionRelation}
    \end{figure}

    \begin{mydefinition}[Thread Pool Reduction Relation]
        We define a \emph{thread pool reduction relation} \tpRedStepSymb according to the rules presented in Figure~\ref{fig:ThreadPoolReductionRelation}. 
        A reduction step has the form
        $$\tpRedStep{\tpVar}{\tidVar}{\nextTpVar}$$
        for a thread ID $\tidVar \in \domOf{\tpVar}$.
    \end{mydefinition}

    \begin{figure}
        \begin{mathpar}
            \RedTPLift
            \and
            \RedTPExit
            \and
            \RedTPThreadTerm
        \end{mathpar}
        \vspace{-0.3cm}
        \caption{Reduction rules for thread pools.}
        \label{fig:ThreadPoolReductionRelation}
    \end{figure}


    Figure~\ref{fig:ExampleProgramCode} illustrates the type of programs we aim to verify.
    The code snippet spawns a new thread which will abruptly terminate the entire program and then busy-waits for the program to be terminated.
    The operational semantics defined above is non-deterministic in regard to when and if threads are scheduled.
    Meanwhile, the presented program only terminates if the exiting thread is eventually scheduled.
    Hence, we need to assume fair scheduling.

    \begin{figure}
            $$
                    \cmdFork[\cmdExit];
                    \cmdLoop[\cmdSkip]
            \vspace{-0.2cm}
            $$
            \caption{Example program with two threads: An exiting thread and one waiting for the program to be abruptly terminated.}
            \label{fig:ExampleProgramCode}
            \vspace{-0.19cm}
    \end{figure}

    \begin{mydefinition}[Reduction Sequence]\label{def:ReductionSequence}
        Let  \tpRSeq{\N} be a sequence of thread pools such that \tpRedStep{\tpVar_i}{\tidVar_i}{\tpVar_{i+1}} holds for all $i\in \N$ and some sequence $(\tidVar_i)_{i\in \N}$ of thread IDs. Then we call \tpRSeq{\N} a \emph{reduction sequence}.
        
    \end{mydefinition}

    Note that according to this definition, reduction sequences are implicitly infinite.

    \begin{mydefinition}[Fairness]\label{def:FairInfiniteReduction}
        We call a reduction sequence \phantom{} \tpRSeq{\N} \emph{fair} iff for all $k \in \N$ and $\tidVar\in\domOf{\tpVar_k}$ 
         there exists $j \geq k$ such that 
        $$
            \tpRedStep{\tpVar_j}{\tidVar}{\tpVar_{j+1}}.
        $$
    \end{mydefinition}

\section{The Logic}\label{sec:Logic}
    
    In this paper, we develop a separation logic to reason about termination of busy-waiting programs.
    Separation logic is designed for reasoning about program resources as well as ghost resources~\cite{Reynolds2002SeparationLA, OHearn2001LocalRA}.
    The latter is information attached to program executions for the purpose of program verification, e.g., a resource tracking how many threads have access to a shared memory location~\cite{Jung2016HigherorderGS}.
    Here, we use ghost resources to track which thread will eventually \cmdExit, i.e., abruptly terminate the entire program.
    
    \paragraph{Obligations \& Credits}
    Remember that \cmdExit terminates all running threads.
    Therefore, in order to modularly reason about program termination we need information about other threads performing \cmdExit.
    
    For this purpose, we introduce two kinds of ghost resources: \emph{obligations} and \emph{credits}.
    Threads holding an obligation are required to perform \cmdExit while threads holding a credit are allowed to busy-wait for another thread to \cmdExit.
    As seen in the next section we ensure that no thread (directly or indirectly) waits for itself.

    We aggregate obligations into obligations chunks, where each obligations chunk collects the held obligations of a single thread.

    \paragraph{Assertions}
    The language of assertions defined in the following allows us to express knowledge and assumptions about held obligations and credits.
    The language contains the standard separating conjunction $\cdot \slStar \cdot$ as well as two non-standard predicates \obsPred and \credit to express the possession of ghost resources.
    (i)~\obs{n} expresses the possession of one obligations chunk containing $n$ \cmdExit obligations; i.e., it expresses that the current thread holds $n$ exit obligations~\footnote{
            As outlined in Section~\ref{sec:FutureWork}, we plan to extend this logic to one where threads are obliged to set ghost signals.
            This makes it necessary to track how many signals remain to be set.
            Hence, we track the number of obligations.
        } .
    (ii)~\credit expresses the possession of an \cmdExit credit that can be used to busy-wait for another thread to \cmdExit.

    \begin{mydefinition}[Assertions]
        Figure~\ref{fig:AssertionsSyntax} defines the set of assertions \AssertionSet.
    \end{mydefinition}
    
    \begin{figure}
        $$
        \begin{array}{r c l} 
            \assVar \in \AssertionSet &::= 
                    &\slTrue ~|~ \slFalse ~|~ \assVar \slStar \assVar  ~|~ 
                    \obs{\natVar} ~|~ \credit\\
            \natVar \in \N
        \end{array}
        $$
        \vspace{-0.5cm}
        
        \caption{Syntax of assertions.}
        \label{fig:AssertionsSyntax}
    \end{figure}

    As we see in Section~\ref{sec:ProofSystem} it is crucial to our verification approach that the \obsPred-predicate captures a full obligations chunk and that this chunk can only be split when obligations are passed to a newly forked thread.
    We represent the information about the held obligations chunks and credits by resource bundles $(\obBagVar, \credVar)$.

    \begin{mydefinition}[Resource Bundles]
        We define the set of \emph{resource bundles} \ResourceBundleSet as
        $$\ResourceBundleSet\ \ :=\ \ \BagsOf{\N} \times \N.$$
        For
        $(\obBagVar_1, \credVar_1), (\obBagVar_2, \credVar_2) \in \ResourceBundleSet$
        we define
        $$
                \begin{array}{l c l}
                        (\obBagVar_1, \credVar_1)\ \tupleCup\
                        (\obBagVar_2, \credVar_2)
                        &:=
                        &(\obBagVar_1 \msCup \obBagVar_2, 
                            \credVar_1 + \credVar_2),
                 \end{array}
        $$
    \end{mydefinition}

    Threads hold exactly one obligations chunk, i.e., resources $(\obBagVar, \credVar)$ with $\cardinalityOf{\obBagVar} = 1$.
    We call such resource bundles \emph{\completeTerm}.
    
    \begin{mydefinition}[Complete Resource Bundles]
        We call a resource bundle $(\obBagVar, \credVar) \in \ResourceBundleSet$ \emph{\completeTerm} if $\cardinalityOf{\obBagVar} = 1$ holds and write \complete{(\obBagVar, \credVar)}.
    \end{mydefinition}
   
    Note that the following definition indeed ensures that the \obsPred-predicate captures a full obligations-chunk.
    Hence, no bundle with one obligations chunk can satisfy an assertion of the form $\obs{n} \slStar \obs{n'}$.     

    \begin{mydefinition}[Assertion Model Relation]
        Figure~\ref{fig:AssertionModelRelation} defines the \emph{assertion model relation} $\assModelsSymb\ \subseteq\ \ResourceBundleSet\, \times\, \AssertionSet$.
        We write
        $$ \assModels{\rbVar}{\assVar}$$
        to express that resource bundle $\rbVar\in \ResourceBundleSet$ models assertion $\assVar \in \AssertionSet$.
    \end{mydefinition}

    \begin{figure}

        $$
        \begin{array}{r c l c l}
            \rbVar &\assModelsSymb &\slTrue
            \\
            \rbVar &\assModelsSymb &\assVar_1 \slStar \assVar_2 
                    &\text{iff}
                    &\exists \rbVar_1, \rbVar_2\in \ResourceBundleSet.\ 
                            \rbVar = \rbVar_1\tupleCup \rbVar_2\\
                    &&&&
                            \wedge\
                            \rbVar_1 \assModelsSymb \assVar_1\ \wedge\ 
                            \rbVar_2 \assModelsSymb \assVar_2
            \\
            (\obBagVar, \credVar) &\assModelsSymb &\obs{\obVar}
                    &\text{iff}
                    &\obVar \in \obBagVar
            \\
            (\obBagVar, \credVar) &\assModelsSymb &\credit
                    &\text{iff}
                    &\credVar \geq 1
        \end{array}
        $$
        \vspace{-0.3cm}
        \caption{Modeling relation for assertions.}
        \label{fig:AssertionModelRelation}
    \end{figure}

\section{Verifying Termination of Busy-Waiting}\label{sec:ProofSystem}
    In this section we present the proof system we propose for verifying termination of programs with busy-waiting for abrupt program exit and illustrate its application.
   Further, we present a soundness theorem stating that every program, which provably discharges all its \cmdExit obligations and starts without credits, terminates.

    \paragraph{Hoare Triples}
    We use Hoare triples \hoareTriple{\htPreConVar}{\cmdVar}{\htPostConVar}~\cite{Hoare1968HoareLogic} to specify the behaviour of programs.
    Such a triple expresses that given precondition \htPreConVar, command \cmdVar can be reduced without getting stuck and if this reduction terminates, then postcondition \htPostConVar holds afterwards.
    In particular, a triple \hoareTriple{\htPreConVar}{\cmdVar}{\slFalse} expresses that \cmdVar does not terminate \emph{normally}, i.e., it either diverges or exits abruptly.

    \paragraph{Ghost Steps}
    When verifying the termination of a program \cmdVar, we consider it to start without any obligations or credits, i.e., \hoareTriple{\noObs}{\cmdVar}{\htPostConVar}.
    Obligation-credit pairs can, however, be generated during so-called \emph{ghost steps}.
    These are steps that exclusively exist on the verification level and only affect ghost resources, but not the program's behaviour~\cite{Jung2018IrisFT, Fillitre2016TheSO}.
    A credit can also be cancelled against an obligation.
    
    \paragraph{View Shift}
    In our proofs, we need to capture ghost steps as well as drawing conclusions from assertions, e.g., rewriting $\htPreConVar \slStar \htPostConVar$ into $\htPostConVar \slStar \htPreConVar$ and concluding \noObs from the assumption \slFalse.
    We ensure this by introducing a view shift relation \viewShiftSymb~\cite{Jung2018IrisFT}.
    A view shift \viewShift{\htPreConVar}{\htPostConVar} expresses that whenever \htPreConVar holds, then either 
    (i)~\htPostConVar also holds or
    (ii)~\htPostConVar can be established by performing ghost steps.
    \lrViewShift{\htPreConVar}{\htPostConVar} stands for
    $\viewShift{\htPreConVar}{\htPostConVar}\, \wedge\,
    \viewShift{\htPostConVar}{\htPreConVar}$.

    \begin{mydefinition}[View Shift]\label{def:ViewShiftRelation}
        We define the view shift relation $\viewShiftSymb\ \subset \AssertionSet \times \AssertionSet$ according to the rules presented in Figure~\ref{fig:ViewShiftRelation}.
    \end{mydefinition}
        
    \begin{figure}
            \begin{mathpar}
                \VSObCredIntro
                \and
                \VSSemImp
                \and
                \VSTrans
            \end{mathpar}
            
            \vspace{-0.3cm}
            \caption{View shift rules.}
            \label{fig:ViewShiftRelation}
    \end{figure}

    Note that view shifts only allow to spawn or remove obligations and credits simultaneously.
    This way, we ensure that the number of obligations and credits in the system remains equal at any time (provided this also holds for the program's initial state).

    \paragraph{Proof Rules}
    We verify program specifications \hoareTriple{\htPreConVar}{\cmdVar}{\htPostConVar}
    via a proof relation \htProvesSymb defined by a set of proof rules.
    These rules are designed to prove that every command \cmdVar, which provably discharges its obligations, i.e.,
    \htProves{\obs{n}}{\cmdVar}{\noObs},
    terminates under fair scheduling.

    \begin{mydefinition}[Proof Relation]\label{def:ProofRelation}
        We define a proof relation \htProvesSymb for Hoare triples \hoareTriple{\htPreConVar}{\cmdVar}{\htPostConVar}
        according to the rules presented in Figure~\ref{fig:proofRules}.
    \end{mydefinition}

    \begin{figure}
        \begin{mathpar}
            \PRFrame
            \and
            \PRExit
            \and
            \PRLoop
            \and
            \PRFork
            \and
            \PRSeq
            \and
            \PRViewShift
        \end{mathpar}
        \vspace{-0.3cm}
        \caption{Proof rules.}
        \label{fig:proofRules}
        \vspace{-0.4cm}
    \end{figure}

    Obligation-credit pairs can be generated and removed via a ghost step by applying \PRViewShiftName plus \VSObCredIntroName.
    The only way to discharge an obligation, i.e., removing it without simultaneously removing a credit, is via rule \PRExitName.
    That is, a discharging program \htProves{\obs{1}}{\cmdVar}{\noObs}, must involve an abrupt \cmdExit at some point.
    
    We can pass obligations and credits to newly forked threads by applying \PRForkName.
    However, note that in order to prove anything about a command \cmdFork[\cmdVar], we need to prove that the forked thread discharges or cancels all of its obligations.

    The only way to justify busy-waiting is via \PRLoopName, which requires the possession of a credit.
    Note that the rule forbids the looping thread to hold any obligations.
    This ensures that threads do not busy-wait for themselves to \cmdExit.
    In other words, we can only prove a triple of the form
    \hoareTriple{\htPreConVar}{\cmdLoop[\cmdSkip]}{\htPostConVar}
    if precondition \htPreConVar expresses the knowledge that another thread is obliged to \cmdExit.

    \paragraph{Example}
    Consider the program
    $\cmdExVar = 
        \cmdFork[\cmdExit];
        \cmdLoop[\cmdSkip]
    $
    presented in Figure~\ref{fig:ExampleProgramCode}.
    It forks a new thread instructed to \cmdExit and busy-waits for it to do so.
    We can verify its termination under fair scheduling by proving 
    \htProves{\noObs}{\cmdExVar}{\noObs}.
    Figure~\ref{fig:VerificationExampleCode} sketches this proof.
    Note that the assumption of fair scheduling is essential, since otherwise we would have no guarantees that the exiting thread is ever executed.

    \begin{figure}
            $$
            \begin{array}{l l}
                    \progProof{\obs{0}}\\
                    \progProof{\obs{1} \slStar \credit}
                            &\proofRuleHint{\PRViewShiftName + \VSObCredIntroName}\\
                    \cmdFork
                            &\proofRuleHint{\PRForkName}\\
                            
                            \quad\progProof{\obs{1}}\\
                            \quad\cmdExit;
                                    &\proofRuleHint{\PRExitName}\\
                            \quad\progProof{\slFalse}\\
                            \quad\progProof{\obs{0}}
                                    &\proofRuleHint{\PRViewShiftName + \VSSemImpName}\\
                    
                    \progProof{\obs{0} \slStar \credit}\\
                    \cmdLoop[\cmdSkip]
                            &\proofRuleHint{\PRLoopName}\\
                            
                            
                    \progProof{\slFalse}\\
                    \progProof{\obs{0}}
                            &\proofRuleHint{\PRViewShiftName + \VSSemImpName}
            \end{array}
            $$
            \vspace{-0.2cm}
            \caption
            {
                    Verification sketch for a program with two threads: 
                    An exiting thread and one busy-waiting for abrupt termination.
                    Applied proof rules are highlighted in violet.
            }
            
            \label{fig:VerificationExampleCode}
            \vspace{-0.3cm}
    \end{figure}

    The following  \hspace{0.02cm}soundness\hspace{0.02cm} theorem\hspace{0.02cm} states that we can prove termination of a program \cmdVar under fair scheduling by proving that it discharging all its \cmdExit obligations,
    i.e.,
    \htProves{\obs{n}}{\cmdVar}{\noObs}.
    By such a proof we verify that no fair infinite reduction sequence of \cmdVar exists.
    That is, the reduction eventually terminates, either abruptly via \cmdExit or normally.
     

    \begin{restatable}[Soundness]{mytheorem}{SoundnessTheorem}
    \label{theo:Soundness}
        Let
        $\htProves{\obs{n}}{\cmdVar}{\noObs}$.
        There exists no fair reduction sequence $(\tpVar_i)_{i\in\N}$ starting with $\tpVar_0 = \setOf{(\tidVar_0, \cmdVar;\contDone)}$ for any $\tidVar_0 \in \N$.
    \end{restatable}

\section{Soundness}\label{sec:Soundness}

    In this section, we sketch the proof of Soundness Theorem~\ref{theo:Soundness} as well as all necessary concepts and lemmas.
    Formal definitions and proofs can be found in the technical report~\cite{Reinhard2020AbruptExitTR}.

    \paragraph{Bridging the Gap}
    According to our proof rules, ghost resources are not static but affected by the commands occurring in a program.
    For instance, forking allows us to pass resources to the newly forked thread and exiting discharges obligations.
    Hence, we annotate threads by resource bundles to obtain annotated thread pools and define annotated versions \rstRedStepSymb and \rtpRedStepSymb of the relations \stRedStepSymb and \tpRedStepSymb.
    We define \rtpRedStepSymb as the union of two reduction relations:
    (i)~\ngrtpRedStepSymb for actual program execution steps and 
    (ii)~\grtpRedStepSymb for ghost steps.
    These relations reflect the interaction between commands and ghost resources as reflected in the proof rules as well as ghost steps realized by view shifts.
    
    Figure~\ref{fig:ExampleStepRules} shows three of the rules we use to define \rstRedStepSymb, \grtpRedStepSymb, \ngrtpRedStepSymb.
    A ghost step performed by \GSObCredIntroName spawns an obligation-credit pair, comparable to a verification ghost step realized by a view shift.
    Reduction rule \RARedSTLoopName requires the looping thread to hold a credit but forbids it to hold any obligation.
    This ensures that busy-waiting threads do not wait for themselves and corresponds to the restriction that proof rule \PRLoopName imposes on looping threads.
    Furthermore, we only allow annotated threads to terminate, 
    i.e., be removed from the thread pool,
    if they do not hold any obligations, as shown by rule \RARedTPThreadTermName.
    Note that in contrast to the plain semantics, reductions in the annotated semantics can get stuck.
    
    The annotated semantics act as an intermediary between high-level reasoning steps, e.g., using obligations to track which thread is going to \cmdExit, and the actual program executions.

    \begin{figure}
            \begin{mathpar}
                    \RARedSTLoop
                    \and
                    \GSObCredIntro
                    \and
                    \RARedTPThreadTerm
            \end{mathpar}
            \caption{Example reduction rules for \ngrtpRedStepSymb and \grtpRedStepSymb.}
            \label{fig:ExampleStepRules}
    \end{figure}

    \paragraph{Interpreting Hoare Triples}
    \, We interpret \phantom{} Hoare triples in terms of a model relation \htModelsSymb and an auxiliary safety relation \isSafe{\rbVar}{\cmdVar}.
    Intuitively, a continuation \contVar is safe under a complete resource bundle \rbVar if \rbVar provides all necessary ghost resources such that the reduction of  $(\rbVar, \contVar)$ does not get stuck.
    We write \isResAnnot{\rtpVar}{\tpVar} to express that \rtpVar is an annotated version of \tpVar, containing the same threads but each equipped with a resource bundle.
    
    \begin{restatable}[Safety]{mydefinition}{defSafety}\label{def:Safety}
        We define the safety predicate $\isSafeName \subseteq \ResourceBundleSet \times \ContSet$ coinductively as the greatest solution (with respect to $\subseteq$) of the following equation:
        $$
        \begin{array}{l}
                \isSafe{\rbVar}{\contVar}\ =\ \complete{\rbVar} \rightarrow\\
                \forall \tpVar, \nextTpVar.\
                \forall \tidVar \in \domOf{\tpVar}.\
                \forall \rtpVar.\\
                \ \
                        \tpVar(\tidVar) = \contVar \wedge
                        \tpRedStep{\tpVar}{\tidVar}{\nextTpVar} \wedge
                        \isResAnnot{\rtpVar}{\tpVar} \wedge
                        \rtpVar(\tidVar) = (\rbVar, \contVar)
                        \rightarrow\\
                \ \
                                \exists \grtpVar.\
                                \exists \nextRtpVar.\
                                        \grtpMultRedStep{\rtpVar}{\tidVar}{\grtpVar} \wedge
                                        \ngrtpRedStep{\grtpVar}{\tidVar}{\nextRtpVar} \wedge
                                        \isResAnnot{\nextRtpVar}{\nextTpVar}\\
                \ \ \phantom{\exists \grtpVar.\ }
                                        \wedge 
                                        \forall (\rbVar^*, \contVar^*) \in 
                                            \rangeOf{\nextRtpVar} \setminus \rangeOf{\rtpVar}.\
                                                    \isSafe{\rbVar^*}{\contVar^*}
        \end{array}
        $$
    \end{restatable}

    \begin{restatable}[Hoare Triple Model Relation]{mydefinition}{defHoareTripleModelRelation}\label{def:HoareTripleModelRelation}
        We define the Hoare triple model relation \htModelsSymb such that
        $$
        \begin{array}{c}
                \htModels{\htPreConVar}{\cmdVar}{\htPostConVar}\\
                \Longleftrightarrow\\
                \begin{array}{l}
                        \forall \frbVar.\
                        \forall \contVar.\
                                (\forall \rbPostVar.\
                                        \assModels{\rbPostVar}{\htPostConVar}
                                        \ \rightarrow\ 
                                        \isSafe{\rbPostVar \tupleCup \frbVar}{\contVar})
                        \\
                        \phantom{\forall \frbVar.\
                                                \forall \contVar.\ }
                                \rightarrow\ 
                                (\forall \rbPreVar.\
                                        \assModels{\rbPreVar}{\htPreConVar}
                                        \ \rightarrow \
                                        \isSafe{\rbPreVar \tupleCup \frbVar}{\,\cmdVar\, ; \contVar})
                \end{array}
        \end{array}
        $$
    \end{restatable}
    
    Every specification \hoareTriple{\htPreConVar}{\cmdVar}{\htPostConVar} we can derive with our proof rules also holds in our model.
    
    \begin{restatable}[Soundness of Hoare Triples]{mylemma}{lemSoundnessHoareTriples}\label{lem:SoundnessHoareTriples}
        Let \htProves{\htPreConVar}{\cmdVar}{\htPostConVar}.
        Then \htModels{\htPreConVar}{\cmdVar}{\htPostConVar} holds.
    \end{restatable}
    \begin{restatable}{proof}{proofLemSoundnessHoareTriples}
        By induction on the derivation of \htProves{\htPreConVar}{\cmdVar}{\htPostConVar}.
    \end{restatable}

    \paragraph{Constructing Annotated Executions}
    Given a fair reduction sequence \tpRSeq{\N} starting with 
    $\tpVar_0 = \setOf{(\tidVar, \cmdVar;\contDone)}$,
    we can construct a fair annotated reduction sequence \rtpRSeq{\N}.
    Our definition of fairness for annotated reduction sequences forbids \rtpRSeq{\N} to perform ghost steps forever.
    
    \begin{restatable}[Fair Annotated Reduction Sequences]{mydefinition}{defFairAnnotatedRedSeq}
            We call an annotated reduction sequence \rtpRSeq{\N} fair iff
            for all $k\in \N$ and $\tidVar \in \domOf{\rtpVar_k}$
            there exists $j \geq k$ such that
            $$\ngrtpRedStep{\rtpVar_j}{\tidVar}{\rtpVar_{j+1}}.$$
    \end{restatable}

    \begin{restatable}{mylemma}{lemModelAllowsAnnotation}\label{lem:ModelRelationAllowsAnnotation}
        Let \htModels{\htPreConVar}{\cmdVar}{\htPostConVar}, \assModels{\rbPreVar}{\htPreConVar} and \complete{\rbPreVar}.
        Furthermore, let \tpRSeq{\N} be fair with
        $\tpVar_0 = \setOf{(\tidVar, \cmdVar;\contDone)}$.
        There exists a fair annotated reduction sequence \rtpRSeq{\N} with
        $\rtpVar_0 = \setOf{(\tidVar, (\rbPreVar, \cmdVar;\contDone))}$.
    \end{restatable}


    As next step, we show that a fair and infinite annotated reduction sequence such as \rtpRSeq{\N} constructed above, must start with an initial credit.
    We do this by analysing the sequence's program order graph 
    $\progOrdGraph(\rtpRSeq{\N})$
    sketched below.
    
    \paragraph{Program Order Graph}
    In the program order graph \phantom{\,} $\progOrdGraph(\rtpRSeq{\N})$
    of
    \rtpRSeq{\N}, every node $i$ represents the $i^\text{th}$ reduction step of the sequence, i.e.,
    \rtpRedStep{\rtpVar_i}{\tidVar_i}{\rtpVar_{i+1}}.
    An edge has the form $(i, \tidVar, n, j)$ and expresses that one of the following hold:
    \begin{itemize}
            \item \rtpRedStep{\rtpVar_j}{\tidVar_j}{\rtpVar_{j+1}} is the first step of a thread forked in step~$i$ or
            \item $j$ is the next index representing a reduction of thread $\tidVar_i$ (in which case $\tidVar_i = \tidVar_j$ holds).
    \end{itemize}
    In both cases, $n$ represents the name of the reduction rule applied in step
    \rtpRedStep{\rtpVar_i}{\tidVar_i}{\rtpVar_{i+1}}.
    We set 0 as the graph's root.

        

    \begin{mylemma}\label{lem:LeafObsEqLeafCreds}
            Let $\progOrdGraph_p$ be a finite 
            sibling-closed
            \ifthenelse
                    {\equal{\thesection}{\ref{sec:Soundness}}}
                    {
                            \!~\footnotemark
                            \footnotetext
                            {
                                    A subgraph is sibling-closed if, for each node $n$ of the subgraph, $n$'s predecessors' successors are also in the subgraph. 
                                    In other words, for each fork step node, either both the next step of the forking thread and the first step of the forked thread are in the subgraph, or neither are.
                            }
                    }
                    {}
            \!\!
            prefix of some program order graph $\progOrdGraph(\rtpRSeq{\N})$
            with $\rtpVar_0 = \setOf{(\tidVar_0, (\rbVar_0, \contVar_0))}$
            for some \completeTerm $\rbVar_0$.
            For all $l \in \leavesOf{\progOrdGraph_p}$ choose $\tidVar_l, \obVar_l, \credVar_l, \contVar_l$ such that \rtpRedStep{\rtpVar_l}{\tidVar_l}{\rtpVar_{l+1}} and 
            $\rtpVar_l(\tidVar_l) = ((\multiset{\obVar_l}, \credVar_l), \contVar_l)$.



            The sum of the obligations held by the threads reduced in the leaves of $\progOrdGraph_p$ equal the sum of the credits held by these threads:
            $$
                    \Sigma_{l \in \leavesOf{\progOrdGraph_p}}\, \obVar_l
                    \ \ = \ \
                    \Sigma_{l \in  \leavesOf{\progOrdGraph_p}}\, \credVar_l
            $$
    \end{mylemma}

    \begin{proof}
            Proof by induction on the size of $\progOrdGraph_p$.
    \end{proof}

    \begin{mylemma}\label{lem:NoFairAnnotatedReductionSequence}
            There are no fair annotated reduction sequences \rtpRSeq{\N} with
            $\rtpVar_0 = \setOf{(\tidVar_0, ((\multiset{\obVar_0}, 0), \contVar))}$
            for any $\tidVar_0, \obVar_0, \contVar$.
    \end{mylemma}
    
    \begin{proof}[Proof Sketch]
            We prove our claim by contradiction.
            Suppose there exits such a reduction sequence \rtpRSeq{\N}.
            Consider the maximal loop-edge-free sibling-closed prefix \progOrdGraphLF of its program order graph.
            Since the \progOrdGraphLF has no loop edges, it is finite. 
            
            Since the reduction sequence \rtpRSeq{\N} is infinite, at least one of the leaves of the prefix \progOrdGraphLF is the predecessor of a loop edge. 
            This leaf  (i.e., the thread reduced in the leaf) must therefore hold a credit but no obligations. 
            By Lemma~\ref{lem:LeafObsEqLeafCreds}, for any finite prefix of \progOrdGraph, the sum of obligations held by the leaves equals the sum of credits held by them.
            Hence, some other leaf $v_o$ must hold an obligation.
            Every leaf is either 
            (i)~the final step of a thread or 
            (ii)~the predecessor of a loop edge.
            Otherwise, by fairness, \progOrdGraphLF would also include the node's successors.
            In both cases (i) and (ii), the leaf is forbidden to hold any obligations.
            However, leaf $v_o$ does hold an obligation.
            This contradicts $v_o$ being a leaf and, hence, the infinite reduction sequence \rtpRSeq{\N} cannot exist.
    \end{proof}
 

    \SoundnessTheorem*
    \begin{proof}
            By Lemmas~\ref{lem:SoundnessHoareTriples}, \ref{lem:ModelRelationAllowsAnnotation},
            \ref{lem:NoFairAnnotatedReductionSequence}.
    \end{proof}

\NewDocumentCommand{\eventThread}{o}
        {\ensuremath{
                t_e
                \IfValueT{#1}{^{#1}}
        }\xspace}
\NewDocumentCommand{\waitThread}{o}
        {\ensuremath{
                t_w
                \IfValueT{#1}{^{#1}}
        }\xspace}        
    
\section{Future Work}\label{sec:FutureWork}
    We are currently formalizing the presented approach and its soundness proof in Coq.

    \paragraph{Ghost Signals}
    We plan to extend the verification approach described in this paper to a verification technique we call \emph{ghost signals}, which allows us to verify termination of busy-waiting for arbitrary events.
    Ghost signals come with an obligation to set the signal and a credit that can be used to busy-wait for the signal to be set.
    Consider a program with two threads: \eventThread eventually performs some event $X$  (such as setting a flag in shared memory) and \waitThread  busy-waits for $X$.
    By letting \eventThread  set the signal when it performs $X$, and thereby linking the ghost signal to the event, we can justify termination of \waitThread's busy-waiting.

            
    
    \paragraph{I/O Liveness as Abrupt Program Exit}
    In concurrent work we encode I/O liveness properties as abrupt program termination following a conjecture of \citet{Jacobs2018ModularTerminationVerification}.
    Consider a \emph{non-terminating} server \serverVar which shall reply to all requests.
    We can prove liveness of \serverVar using the following methodology:
    \begin{itemize}
        \item For some arbitrary, but fixed $N$, assume that responding to the $N^{\text{th}}$ request abruptly terminates the whole program.
        \item Prove that the program always terminates.
    \end{itemize}

    \begin{figure}
          \begin{center}
                \includegraphics[scale=0.86]{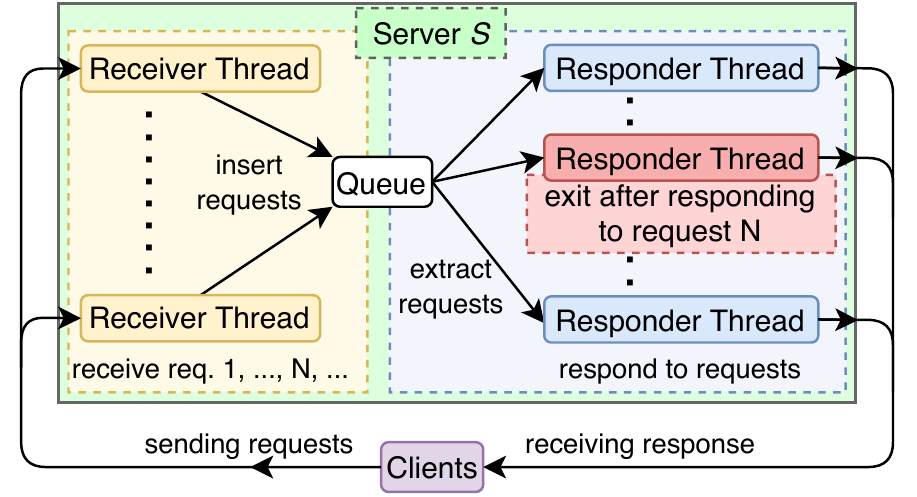}
            \end{center}
        
        \vspace{-0.3cm}
        \caption
            {Server \serverVar receiving and replying to requests.
             Threads communicating via shared queue.}
        \label{fig:ServerMultipleThreads}
        \vspace{-0.4cm}
    \end{figure}

    One can combine this approach with the one of the present paper to verify liveness of a server where multiple threads independently receive and handle requests. 
    Using a \emph{prophecy variable}~\cite{Jung2020POPLProphecyVariables}, one can determine ahead of time which thread will receive the exiting request. 
    The other threads can then be seen as busy-waiting for this thread to exit.

\paragraph{Combination}

    We plan to combine the two approaches sketched above and conjecture that the combination will be expressive enough to verify liveness of programs such as the server \serverVar presented in Figure~\ref{fig:ServerMultipleThreads}.
    It runs a set of receiving and a set of responding threads communicating via a shared queue.
    The responding threads busy-wait for requests to arrive in the queue.
    In order to verify liveness of \serverVar, we need to show that some thread
    eventually abruptly terminates the program by acquiring and responding to the $N^\text{th}$ request.
    This requires us to prove that the responding threads' busy-waiting for requests to arrive in the shared queue terminates.
    
    
    In order to demonstrate the approach's usability, we plan to implement it in VeriFast~\cite{Jacobs2011Verifast} and prove liveness of \serverVar.
    

\section{Related Work}\label{sec:RelatedWork}
\NewDocumentCommand{\TadaLive}{}{TaDA Live\xspace}
\NewDocumentCommand{\Lili}{}{LiLi\xspace}

    \citet{Liang2016LiliAPL, Liang2017LiliProgressOC} propose \Lili, a separation logic to verify liveness of blocking constructs implemented via busy-waiting.
    In contrast to our verification approach, theirs is based on the idea of contextual refinement.
    \Lili does not support forking nor structured parallel composition.

    \citet{DOsualdo2019arxivTaDALive} propose \TadaLive, a separation logic for verifying termination of busy-waiting.
    This logic allows to modularly reason about fine-grained concurrent programs and blocking operations that are implemented in terms of busy-waiting and non-blocking primitives.
    It uses the concept of obligations to express thread-local liveness invariants, e.g., that a thread eventually releases an acquired lock.
    \TadaLive is expressive enough to verify CLH and spin locks.
    The current version supports structured parallel composition instead of unstructured forking.
    Comparing their proof rules to ours, it is fair to say that our logic is simpler. 
    Of course, theirs is much more powerful. 
    We hope to be able to extend ours as sketched above while remaining simpler.

\section{Conclusion}\label{sec:Conclusion}
    In this paper we proposed a separation logic to verify the termination of programs where some threads abruptly terminate the program and others busy-wait for abrupt termination.
    We proved our logic sound and illustrated its application.

    Abrupt termination can be understood as an approximation of a general event.
    We outlined our vision on how to extend our approach to verify the termination of busy-waiting for arbitrary events.
    We have good hope that the final logic will be as expressive as \TadaLive proposed by \citet{DOsualdo2019arxivTaDALive} while remaining conceptually simpler.

    Further, we sketched our vision to combine this extended work with concurrent work on the encoding of I/O liveness properties as abrupt termination.
    We illustrated that this combination will be expressive enough to verify liveness of concurrent programs where multiple threads share I/O responsibilities.

\bibliographystyle{plainnat}
\bibliography{bibliography}

\end{document}